%% file: esterarxi.tex
\documentclass[twocolumn,aps,prl,showpacs,10pt]{revtex4-1} 

\usepackage{graphicx}
\usepackage{amsfonts}
\usepackage{amsthm}
\usepackage{amsmath}
\usepackage{amssymb}
\usepackage{graphicx}
\usepackage{color}
\usepackage{dsfont}
\usepackage[normalem]{ulem}

\newif\ifarxi   
\arxitrue  

\theoremstyle{definition}
\newtheorem{definition}{Definition}
\theoremstyle{plain}
\newtheorem{theorem}[definition]{Theorem}

\theoremstyle{definition}

\theoremstyle{plain}

\newcommand{\be}{\begin{eqnarray}}
\newcommand{\ee}{\end{eqnarray}}
\newcommand\ba{\begin{array}}
\newcommand\ea{\end{array}}

\newcommand{\bra}[1]{\langle #1|}

\newcommand{\ket}[1]{| #1 \rangle }

\newcommand{\state}[1]{{{\mathcal S}(#1)}}
\newcommand{\alg}[1]{{\mathfrak {#1}}}
\newcommand{\hs}[1]{{\mathcal #1}}

\newcommand{\ml}{l\kern-0.035cm\char39\kern-0.003cm}
\newcommand{\mt}{t\kern-0.035cm\char39\kern-0.003cm}
\newcommand{\md}{d\kern-0.035cm\char39\kern-0.003cm}

\def\hh{{\mathcal H}}
\def\mm{{\mathcal M}}

\def\bb{{\mathcal E}}
\def\cE{{\mathcal E}}

\def\id{{\mathbb I}}

\def\tr{{\rm Tr}}

\begin{document}
\title{Process estimation in presence of time-invariant memory effects}
\author{Tom\'a\v s Ryb\'ar$^1$ and M\'ario Ziman$^{2,3}$}
\address{
$^{1}$
Institut f\"ur Theoretische Physik, Leibniz Universit\"at Hannover, Appelstr. 2, 30167 Hannover, Germany\\
$^{2}$Institute of Physics, Slovak Academy of Sciences, D\'ubravsk\'a cesta 9, 845 11 Bratislava, Slovakia \\
$^{3}$Faculty of Informatics, Masaryk University, Botanick\'a 68a, 60200 Brno, Czech Republic 
}
\begin{abstract}
Any repeated use of a fixed experimental instrument is subject to memory effects. We design an estimation method uncovering the details of the underlying interaction between the system and the internal memory without having any experimental access to memory degrees of freedom. In such case, by definition, any memoryless quantum process tomography (QPT) fails, because the observed data sequences do not satisfy the elementary condition of statistical independence. However, we show that the randomness implemented in certain QPT schemes is sufficient to guarantee the emergence of observable "statistical" patterns containing complete information on the memory channels. We demonstrate the algorithm in details for case of qubit memory channels with two-dimensional memory. Interestingly, we found that for arbitrary estimation method the memory channels generated by controlled unitary interactions are indistinguishable from memoryless unitary channels.
\end{abstract}
\maketitle

\ifarxi\section*{Introduction}\fi

Repeatability of experiments is one of the main conceptual paradigms of
modern science although its meaning during the history has evolved. 
In particular, the quantum experiments are not repeatable in a strict 
sense of individual observations (e.g. no one knows whether a given 
photon passes the polarizer, or not), however, the repeated runs of
such experiments exhibit repeatable statistical patterns (e.g. the 
fraction of photons passing the polarizer is fixed). In other words, 
quantum theory does not give a clear conceptual meaning 
(in sense of repeatability) to individual outcomes, but rather 
to numbers represented by averages and probabilities. 

Therefore, the interpretation of quantum experiments is intimately related 
with our understanding of probabilities, especially with the question  
whether the observed frequencies are really the probabilities occuring 
in theoretical models of the experiments. In any case the repeatability 
of statistical features assumes
that individual runs of the experiment are independent. In theory this means
that each run of the experiment is performed with "fresh" apparatuses 
(under exactly the same conditions), however,
in practise, we do not really employ a new apparatus every time 
the experiment is run. Instead, it is implicitly assumed that the 
internal relaxation processes are sufficiently fast to refresh the whole 
experimental setup. But is such assumption justified?

Consider an experiment in which a quantum particle is sent 
through a quantum channel. While the particle is transfered it interacts
with the degrees of freedom of the channel. According to quantum theory
these interactions are described by Schr\"odinger equation and results
in a unitary transformation of the joint particle-channel system. 
As a result both the particle and the channel are disturbed by this
interaction and the disturbances depends on their original 
characteristics. Consequently, the repeated use of the same channel 
device is not independent of the previous uses, thus, the induced
particle transformation will be typically different. If this is the
case we say the channel exhibits memory effects. Let us stress that 
all the relaxation processes can be incorporated into this unitary model
by extending the size of the memory. 

Indeed, suppose the channel is just "delaying" the transfer of the
particles, i.e. its n$th$ output equals $(n-1)$th input (first
output is set to be in some fixed state). In this case, the
uses are clearly not independent. This can be demonstrated if one's goal
is to estimate the parameters of the quantum process assuming the channel
devices are memoryless. Then different (equivalent in memoryless case) 
estimation procedures could lead to different conclusions. In particular,
if the channel action is tested in an "ordered" fashion, i.e. we first 
analyze how the state $\varrho_1$ is transformed to $\varrho_1^\prime$, 
then $\varrho_2$ to $\varrho_2^\prime$, etc., then any delay vanishes
in the statistical analysis and we must conclude the transformation is
noiseless, i.e. $\varrho\mapsto\varrho^\prime=\varrho$. However,
if the channel is tested in a "random" fashion, i.e. in each run
a random test state is used, then for each fixed input $\varrho$
the output state $\varrho^\prime$ is a fixed state $\varrho_0$ being 
the average input test state, thus, the channel is recognized as 
the maximal noise and therefore not very useful for the transfer 
\emph{per se}.

In the described case, the action of the memory is quite simple and 
when cleverly used, this memory device can be used to transfer information 
in a noiseless way \cite{kretschmann2005}. But how to find out the 
action if the interaction is not known in advance? How to proceed in order 
to detect such memory behavior and finally exploit the memory for our 
purposes? Exactly 
these questions will be addressed in this Letter. It is
organized as follows. We start with introducing all the necessary concepts
and tools. Then we formulate theorems allowing us to design the estimation
algorithm. Finally we will illustrate in details the algorithm for the simplest
possible case of qubit channels with a two-dimensional memory.

\begin{figure}
 \includegraphics[scale=0.35]{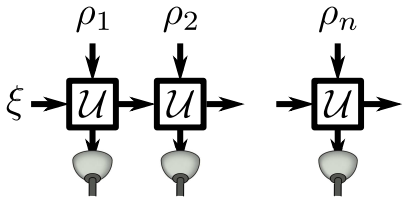}
 \caption{Repeated uses of time-invariant memory process identified with the unitary channel $\cal{U}$ describing the interaction between the device inputs $\varrho_j$ and experimentally inaccessible memory degrees of freedom initially in an unknown state $\xi$.
\label{fig_memchan}}
\end{figure}

\ifarxi\section*{Preliminaries}\else \emph{Preliminaries. }\fi
States $\varrho$ of quantum systems are identified with the set of 
density operators $\state{\hh}$ being positive linear operators
on a Hilbert space $\hh$ of a unit trace. Measurement apparatuses $M$
are described by positive operator valued measures (POVM) being
a set of positive operators $E_1,\dots,E_m$ such that $O\leq E_j\leq I$ 
(called \emph{effects}) and $\sum_j E_j=I$. Each measurement outcome
is associated with exactly one effect and we write $E_k\in M$ if
$E_k$ is an effect associated with one of the outcomes M. 
Quantum channels $\cE$ 
describing the memoryless processes are identified with completely positive 
trace-preserving linear maps defined on the set of traceclass 
operators ${\cal T}(\hh)$. In particular, $\cE:\state{\hh}\to\state{\hh}$.
If $\cE(\varrho)=U\varrho U^\dagger$ for some unitary operator $U$, then
we say the channel is unitary and we denote it by $\cal{U}$. 
Due to Stinespring theorem any channel can be understood as a result
of a unitary interaction between the system and some (initially factorized) 
memory, i.e.  $\cE(\varrho)=\tr_{\cal{M}}[\cal{U}({\xi\otimes\varrho})]$, where $\xi$
is the initial state of the memory and 
$\cal{U}:\state{{\mm\otimes\hh}}\to\state{{\hh\otimes\mm}}$ 
(for more details see for instance Ref.~\cite{heinosaari2013}). 

When modelling (see Refs.~\cite{kretschmann2005,rybar_aps}) 
the experiment with memory process device (used repeatadly) 
we will assume that its action is described by a fixed unitary channel 
and includes all the relaxation processes of the memory. Also we assume
that we do not have any access to memory degrees of freedom, thus, when
we want to learn something about the underlying process $\cal{U}$ we can
only manipulate the system, eventually employing some ancilliary systems 
and devices. The experiment gives raise to a sequence of channels 
$\cE^{{(1)}},\dots,\cE^{{(n)}}$ for $n$ uses of the device defined
as follows
\be
\nonumber
\cE^{{(j)}}(\varrho^{(j)})={\rm tr}_{\mm}[{\cal{U}}^{{(j)}}(\xi\otimes\varrho^{(j)})]\,,
\ee
where ${\cal{U}}^{{(j)}}={\cal{U}}_j\cdots {\cal{U}}_1$ is the $j$-fold concatenation of ${\cal{U}}_{{k}}$ and $\varrho^{(j)}$ is the joint input state describing first $j$ uses of the memory device. The unitary channel ${\cal{U}}_k$ acts as ${\cal{U}}$ on the memory and the $k$th system and trivially elsewhere.

Let us stress that by construction the sequence of processes 
$\cE^{{(1)}},\dots,\cE^{{(n)}}$ is causal ($j$th output does not depend on $k$th input 
for $k>l$), thus, $\cE^{{(j)}}(\varrho^{(j)})={\rm tr}_{j+1}[\cE^{{(j+1)}}(\varrho^{(j+1)})]$. 
Indeed, due to seminal paper by Kretschmann and Werner~\cite{kretschmann2005} 
every causal memory process 
can be represented as a concatenation of unitary channels describing
the sequence of interactions between the memory and systems. 
In our case the memory channel is also time-invariant 
(see Fig.~\ref{fig_memchan}), i.e. the unitary channels applied in
concatenations coincide.

\ifarxi\section*{Quantum process tomography}\fi
\emph{Quantum process tomography} (QPT) is any processing of experimental
data identifying uniquely an unknown memoryless quantum channel 
\cite{paris2004,chuang1997}. It is known to be a complex task, 
however, under certain assumptions 
it can be efficiently applied also to large systems \cite{gross2010,mahler2013,schmiegelow2011,silva2011,cramer2010,flammia2010} and even the accuracy can be assessed \cite{sugiyama2013,kohout2010,kohout2012,christandl12}. 
QPT deals with a scenario when an experimenter is given an unknown 
input-output black box $\bb$. In each run of the experiment he
prepares some test state $\varrho$ and performs a measurement $M$, thus, 
he choses the setting $x=(\varrho,M)$, and, records the outcome 
$E_k$, where $E_k\in M$. Let us denote by 
$X=\{(\varrho_x,M_x)\}_x$ the set of all possible settings. 
The measurement $M_x$ is described by effects $E_{xk}$ and 
by $N_x$ the number of times the setting $x$ was chosen, 
i.e. $\sum_x N_x=n$. In each run of the experiment we observe 
an event $x_k=(\varrho_x,E_{xk})$ saying
that the setting $x=(\varrho_x,M_x)$ was used and the outcome $E_{xk}$
is recorded. The conditional probability of observing the event 
$x_k$ is given by 
formula $p(x_k|\cE)=q_x\tr{E_{xk} \cE[\varrho_x]}$,
where $q_x=N_x/n$ describes the frequency of the setting $x$. 
For suitable choice of $X$  this probability distribution
$p(x_k|\cE)$ enables us to reveal the identity of the channel $\cE$.
Conceptually the simplest example consists of a linearly independent 
collection of test states $\{\varrho_x\}_x$ and a fixed state tomography
measurement $M$ (the same for each $x$).

Clearly, the ordering of the settings $x_1,\dots,x_n$ is irrelevant 
for QPT and only their fraction is needed. However, this is true 
only if the condition of memoryless 
channel are met, i.e. when a fresh copy of the channel is used each 
time the experiment is made. Otherwise the QPT procedure may lead to wrong
conslusions. Suppose we have tested a communication channel (the delaying channel from the Introduction) using
the well-ordered sequence of settings and find out the transfer 
is just perfect, thus, we use it to built a noiseless worldwide 
communication network. However, the communication itself is quite 
far from well-ordered sequence of symbols. It is much closer to a random 
one and for such the considered communication device does not work at all, 
hence, the seemingly "perfect" network fails dramatically. On the other
side, the usage of random sequence of settings leads to a conclusion that
the communication device is of no use. But this is also not true, 
because shifting the outputs by one results in perfect transmission.

\ifarxi\section*{Formulation of the problem}\else \emph{Formulation of the problem. }\fi 
The goal is to capture the underlying dynamics, i.e. the interaction
$\cal{U}$ and the memory state $\xi$. However, as we have only a single copy 
of the state $\xi$, learning any nontrivial information on $\xi$ 
is forbidden by the no-cloning theorem \cite{scarani_nocloning}. 
Moreover, not all the parameters of $\cal{U}$ are accessible within 
our model, too. In particular, the output of the memory channel given by 
$(\cal{U},\xi)$ is the same as of 
 $(({\cal{I}}\otimes{\cal{V}_M}){\cal U}({\cal{V}}^{-1}_{\cal{M}}\otimes{\cal{I}}), {{\cal{V}}_{\cal{M}}}\xi{{\cal{V}}^{-1}_{\cal{M}}})$, for some unitary $\cal{V_M}:\mm\mapsto\mm$. In conclusion, our goal
is to estimate $\cal{U}$ modulo this freedom under the condition that
the initial state of the memory is unknown and the memory is experimentally
inaccessible. 

Before we proceed let us stress that (just like in the memoryless case) 
we are able to predict probabilities, however, by construction our experiments 
cannot be repeated in the statistical sense, hence, the standard
tools and methods of statistical analysis are simply inapplicable. 
In full generality of the problem we are free to choose the input 
state for a given number $n$ of uses of the device, we can choose
the output measurements and we may also employ some ancilla. 

\ifarxi\section*{Controlled unitary interactions}\else \emph{Controlled unitary interactions. }\fi In this example we will show
a family of memory channels, for which the freedom in the estimation 
of the interaction $\cal U$ is much larger. We say the interaction is
controlled unitary, if it can be written in the following form 
\cite{silent_swap}
${\cal{U}}^{\rm ctrl}=\sum_l \ket{l}_{\hs{M}}\bra{l}\otimes {\cal{V}}_l$, where
${\cal{V}}_l$ are arbitrary unitary channels defined on the system and vectors
$\ket{l}$ form an orthonormal basis of the memory Hilbert space.
\begin{theorem} \label{theorem:control_u}
The memory device induced by a controlled unitary interaction 
 ${\cal{U}}^{\rm ctrl}$ is indistinguishable from a memoryless unitary device.
\end{theorem}
\begin{proof}
Suppose $\varrho^{(n)}$ is the joint state of $n$ inputs and let $\xi_{\mm}$
be the initial state of the memory. Then
$\varrho^{(n)\prime}=\sum_l q_l {\cal{V}}_l^{\otimes n} (\varrho^{(n)})$ 
with $q_l=\langle l|\xi|l\rangle$. Suppose $E$ is an effect on $n$ outputs 
such that $p_E({\cal{U}}^{\rm ctrl})
={\rm tr}[E{\cal{U}}^{\rm ctrl}_n(\varrho^{(n)})]>0$. 
Then for the same input state $\varrho^{(n)}$ also 
$p_E({\cal{V}}_l)={\rm tr}[E{\cal{V}}_l^{\otimes n}(\varrho^{(n)})]>0$ 
for some $l$, thus for any test state the observation of the 
individual outcome $E$ cannot be used to distinguish ${\cal{U}}^{\rm ctrl}$ 
from ${\cal{V}}_l$ (for a suitable $l$).
\end{proof}

In other words, any estimation procedure for this
class of channels results randomly (with probability $q_l$) in one 
of the unitaries ${\cal{V}}_l$.

\ifarxi\section*{Estimation algorithm}\else \emph{Estimation algorithm. }\fi  The algorithm we are going to explain
is based on QPT method with randomly chosen settings (see Fig.~\ref{fig_method}). In particular, 
in each run of the experiment
the setting $x=(\varrho_x,M_x)$ is selected indepentently with the 
probability $q_x$. Let us remind that among $n$ uses of the channel
approximatively $N_x\approx q_x n$ times the setting $x$ is selected. 
Denote by $N_{xk}$ be the number of occurences of the event
$(\varrho_x,E_{xk})$ (with $E_{xk}$ being the effect observed
in the measurement $M_x$) and define a number 
$\tilde{p}(k|x)=N_{xk}/N_x$ (playing the role of conditional
probabilities in case of QPT). The following theorem provides
the statistical interpretation of this number.

\begin{theorem}\label{theo:interpretation} 
If QPT is implemented with randomly chosen settings, then for all
settings $x$ there exist a state of the memory 
$\overline{\xi}\in\state{\hs{M}}$ such that
\be
\nonumber
\lim_{n\to\infty}\tilde{p}(k|x)=p(x_k|x)\equiv 
{\rm tr}[{\cal{U}}(\overline{\xi}\otimes\rho_x)(E_{xk}\otimes I_{\hs{M}})]\,.
\ee
Consequently, we may treat $p(x_k|x)$ as the conditional probability 
$p(x_k|\cE)$ for a \emph{average} channel $\cE(\varrho)
={\rm tr}_{\mm}[{\cal{U}}(\overline{\xi}\otimes\varrho)]$ induced
by the state $\overline{\xi}$, hence, QPT reconstruction results 
in the memoryless channel $\cE$.
\end{theorem}
\begin{proof}
Let us denote by $\xi_j$ the state of the memory before $j$th run
of the experiment leading to observation of some effect $E_{x_j k}$.
During the algorithm the memory system undergoes a sequence of 
transformations $\xi\equiv\xi_1\mapsto\xi_2\mapsto\cdots\mapsto\xi_n$
Denote by $\alg{S}$ the set of all states $\{\xi_j\}_j$ occuring in the
sequence and by $\alg{S}_x$ a subset of $\alg{S}$ for which the setting 
$x$ was used.  
Consider a partitioning of $\state{\hs{M}}$ into mutually exclusive subsets
$\{\hs{X}_\mu\}_\mu$, i.e. $\hs{X}_{\mu}\cap\hs{X}_{\nu}=\emptyset$
and $\state{\hs{M}}=\bigcup_{\mu}\hs{X}_{\mu}$. Define 
$p(\hs{X}_{\mu})=|\alg{S}\cap\hs{X}_\mu|/n$ and 
$p_x(\hs{X}_{\mu})=|\alg{S}_x\cap\hs{X}_\mu|/|\alg{S}_x|$ determining the frequency
of the memory state being from the subset $\hs{X}_\mu$ and the
frequency being from $\hs{X}_\mu$ conditioned on the settings $x$, 
respectively. As the choice of the setting $x$ is random the states 
$\xi_j\in\hs{X}_\mu$ are distributed between the sets $\alg{S}_x$ 
at random with probability $q_x$, hence, the subset $\alg{S}_x$ 
is a random sample of $\alg{S}$. Formally 
$|\alg{S}_x\cap\hs{X}_\mu|\approx q_x |\alg{S}\cap\hs{X}_\mu|$ 
for large $n$. Consequently, for all $x$ we obtain the relation
$p_x(\hs{X}_{\mu})\approx p(X_\mu)$, i.e. for any partitioning
the conditional distribution $p_x(\hs{X}_\mu)$ is (in the limit
of large $n$) independent  of the initial settings $x$. In other words,
whatever initial setting is used the average memory state $\overline{\xi}_x$ 
is fixed and $\overline{\xi}_x=\overline{\xi}$. Therefore, 
for each $x$ the observed transformation is $\varrho_x\mapsto\varrho_x^\prime
=(1/n)\sum_{\alg{S}_x}{\rm tr}_\mm {{\cal{U}}(\xi_l\otimes\varrho_x)}
\equiv\cE(\varrho_x)$ with 
$\cE(\varrho_x)={\rm tr}_\mm {{\cal{U}}(\overline{\xi}\otimes\varrho_x)}$.
\end{proof}

Note that a trivial implication of this result is that the average channel on $n$ subsequent inputs reads $\cE_{n}(\rho^{(n)})=\tr_{\mm}[{\cal U}^{(n)}(\overline{\xi}\otimes\rho^{(n)})]$ and corresponds to the probabilities of $n$ joint events. 
This theorem enables us to interpret the result of any QPT method 
with randomly chosen settings, however, it does not tell us what 
the generating state $\overline{\xi}$ is. When a channel $\cE$ is 
reconstructed, then we know the interaction $\cal{U}$
is one of its dilations. The following theorem provides a tool to 
determine the average state $\overline{\xi}$.

\begin{theorem}
\label{th:3}
The average state $\overline{\xi}$ is a fixed point of the
channel ${\cal{C}}(\xi)={\rm tr}_{\hh}[{\cal{U}}(\xi\otimes\overline{\varrho})]$,
where $\overline{\varrho}=\sum_x q_x \varrho_x$ is the average test state, 
i.e. ${\cal{C}}(\overline{\xi})=\overline{\xi}$.
\end{theorem}
\begin{proof} 

Let $f$ be the measure on $\state{\mm}$ characterizing the distribution of states in the set $\alg{S}$ for large $n$, i.e. $\int_{\hs{X}_{\mu}}{\rm  d}f(\xi)\approx|\hs{X}_{\mu}|$ and $\overline{\xi}\approx\int_{\state{\mm}}{\rm d} f(\xi)\xi$. Given that state $\xi$ enters collision with state $\rho_x$ and the measured output is $E_{xk}$, the exiting state of memory is $\xi_{\rm out}={\cal I}_{xk}[\xi]/\tr({\cal I}_{xk}[\xi])$ where ${\cal I}_{xk}[\xi]=\tr_{\hh}[{\cal U}(\xi\otimes\rho_x)(E_{xk}\otimes I_{\mm})]$. The probability of the event $(\rho_x,E_{xk})$ is $q_x\tr({\cal I}_{xk}[\xi])$ where $q_x$ is the probability of setting $(\rho_x, M_x)$. The average over $\alg{S}$ can be expressed as the average over exiting states $\xi_{\rm out}$ (for inputs $\xi$ distributed according to $\mu$), thus,
\begin{equation}
\nonumber
  \overline{\xi}=\sum_{xk}\int_{\state{\mm}}{\rm d} f(\xi)\frac{q_x\tr({\cal I}_{xk}[\xi]){\cal I}_{xk}[\xi]}{\tr({\cal I}_{xk}[\xi])}={\cal{C}}(\overline{\xi})\,.
\end{equation}
\end{proof}
In conclusion, if the mapping $\cal{C}$ has a unique fixed point
$\xi_0$, then $\overline{\xi}=\xi_0$ for large $n$. 

\begin{figure}
 \includegraphics[scale=0.35]{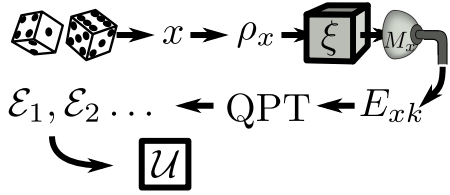}
 \caption{Schematic illustration of the estimation method. Setting $x$ is chosen at random, and outcome $E_{xk}$ is observed. Collecting this data and performing process tomography (QPT) yields a family of channels ${\cal E}_n$ on $n$ subsequent inputs. From these channels the interaction $\cal U$ is determined up to local unitary rotation of the memory system.
\label{fig_method}}
\end{figure}

\ifarxi\section*{Qubit memory channel with two-dimensional memory}\else \emph{Qubit memory channel with two-dimensional memory. }\fi  
In this part we will sketch \cite{supplement} how the described QPT method with randomized inputs can be used to determine the parameters of the unitary interaction $\cal U$ in the simplest case when both the system and the memory are two-dimensional, hence, represented by qubits.

A general two-qubits unitary operator 
can be written in the form \cite{2x2unitary decomposition}
$U=(W_{2}\otimes V_{2})D(\vec{\alpha})(W_{1}\otimes V_{1})$
with $D(\vec{\alpha})$ defined as
\be
\label{eq:Dalpha}
D(\vec{\alpha})=
\exp{\frac{1}{2}\sum_j\alpha_j\sigma_j\otimes\sigma_j}
\ee
where $\sigma_j$ are Pauli operators and 
$0\leq|\alpha_z|\leq\alpha_y\leq\alpha_x\leq\pi/2$. 
Due to equivalence of the memory channels we choose to 
set $W_2=I$, hence, the task is to determine 
$W_1, V_1, V_2$ and $\alpha_i$.

We will consider that QPT consist of test states satisfying 
$\sum_x q_x \varrho_x=\frac{1}{2}I$ and of a fixed informationally complete
measurement (i.e. $M_x=M$ for all $x$). Under this assumptions the channel on memory $\cal{C}$ is unital, thus, having the complete 
mixture as one of the fixed points.

It turns out that for two qubit unitaries the channel $\cal C$ has either a unique fixed point or the interaction is controlled unitary. Consequently, when the QPT estimation with randomized inputs results in a unitary channel, we conclude that the interaction is ${\cal{U}}^{\rm ctrl}$. If the estimated channel is not unitary, then we know that the average state of memory is $\overline{\xi}=I/2$, thus, in Bloch sphere representation the channel $\cal E$ reads ${\cal E}=R_2CR_1$ where $R_i$ are orthogonal rotations corresponding to unitaries $V_i$, respectively, and $C$ is a diagonal matrix $C={\rm diag}[1,c_2 c_3, c_1 c_3, c_1 c_2]$ with $c_l=\cos \alpha_l$. The decomposition ${\cal E}=R_2CR_1$ is obtained by performing singular value decomposition of the Bloch sphere representation of $\cal E$, while taking care that $\det R_i=1$, i.e. they are proper rotations. This is always possible since due to our parametrization the $c_ic_k$ are all positive. From this, one easily gets $|\alpha_i|$. It remains to find the sign of $\alpha_z$ and the unitary $W_1$. This can be achieved \cite{supplement} by estimating the channel on two subsequent inputs ${\cal E}_{2}(\rho_a\otimes \rho_b)=\tr_{\mm}[{\cal{U}}^{(2)}(\overline{\xi}\otimes\rho_a\otimes\rho_b)]$ which we can obtain from the same data set. The method was numerically verified on randomly generated unitaries $\cal{U}$ with random initial memory states. In all cases the method succeeded in correct reconstruction.

\ifarxi \section*{Summary} \else \emph{Summary. }\fi 
We have proposed an estimation method for 
estimating the underlying system-memory interaction $U$ generating 
the memory channel assuming that this interaction is 
time-invariant. The algorithm is based on a random implementation of 
arbitrary memoryless quantum process tomography (QPT) procedure. We proved 
that arbitrary memoryless QPT (implemented with random settings)
results in some memoryless channel $\cE$ with a dilation being the 
system-memory interaction $U$. Moreover, when the average state of the 
memory (during QPT) is known, then the correct identification 
of the interaction (among the unitary dilations of $\cE$) is possible.
We proved this happens when the average testing state is chosen to be 
the complete mixture and the average concurrent channel is strictly 
contractive. In particular, in this case the average memory channel
is the complete mixture as well, hence, the reconstructed 
channel $\cE$ is necessarily unital. The reconstruction method is illustrated
for qubit memory channels with two-dimensional memories.  
It can be extended for systems and memories of arbitrary size, however,
it is an open question what size of concatenation 
$\cE_{n}$ is sufficient for completing the estimation of $\cal U$.

The presented estimation method is universal, however, it is neither 
the most general one, and likely nor the optimal one. We believe that 
our abilities to treat the concept of memory channels in experiments 
provide us with better understanding and control of quantum apparata and, 
therefore, they are not only of a deep foundational interest, 
but have direct application. Conceptually, this work is challenging 
our understanding and interpretation of elementary scientific tools: the
repeatability and the probability. Practically, the problem is 
intimately related to a single-copy estimation of matrix product states 
and the results may be applied for characterization of Hamiltonians 
of single-copy many-body system. In particular, the sequence of repeated 
measurements on the subsystem followed by system's evolution 
(for a fixed time interval) is covered by the considered 
memory channel model. 

\ifarxi \section*{Acknowledgements} \else \emph{Acknowledgements.} \fi
This work has been supported by EU integrated project SIQS and
APVV-0646-10 (COQI). T.R. acknowledges the support of ERC grant DQSIM.
M.Z. acknowledges the support of GA\v CR project P202/12/1142 and RAQUEL.

\ifarxi \section*{Appendix: Estimation of qubit memory channel with two-dimensional memory system} \input{2x2ex.tex} \fi

\end{document}

%% file: 2x2ex.tex
In this example we will assume
both the system and the memory are two-dimensional, hence, 
represtented by qubits. In such case a general unitary operator
can be written in the form \cite{2x2unitary decomposition}
$U=(W_{2}\otimes V_{2})D(\vec{\alpha})(W_{1}\otimes V_{1})$
with $D(\vec{\alpha})$ defined as
\be
\label{eq:Dalpha}
D(\vec{\alpha})=
\exp{\frac{1}{2}\sum_j\alpha_j\sigma_j\otimes\sigma_j}
\ee
where $\sigma_j$ are Pauli operators and 
$0\leq|\alpha_z|\leq\alpha_y\leq\alpha_x\leq\pi/2$. 
Due to equivalence of the considered memory channels we choose to set $W_2=I$.
Further, we will consider that QPT consist of test states satisfying 
$\sum_x q_x \varrho_x=\frac{1}{2}I$ and of a fixed informationally complete
measurement (i.e. $M_x=M$ for all $x$). Under this assumptions the 
channel ${\cal C}(\xi)=\sum_x q_x \tr_{\hh}[{\cal U}(\xi\otimes \rho_x)]$ is unital, thus, having the complete 
mixture as one of the fixed points.

\emph{Unique fixed point.} Let us start with the case when the fixed 
point $\frac{1}{2}I$ is unique, i.e. due to Theorem \ifarxi \ref{th:3} \else 3 (see the main text) \fi
the average memory state reads $\overline{\xi}=\frac{1}{2}I$.
Consequently, the reconstructed channel reads
\be
\cE_1(\rho)=\tr_{\hs{M}}(U(\overline{\xi}\otimes\varrho)U^{\dagger})
=V_2\cE_D(V_1\rho V_1^{\dagger})V_2^{\dagger}\,,
\ee
where $\cE_D$ stands for the channel induced by unitary evolution 
$D(\vec{\alpha})$. In the Bloch sphere representation it is represented
by the diagonal matrix
\be\label{eq:superoper}
\cE_D={\rm diag}[1,C]={\rm diag}[1, c_2 c_3, c_1 c_3, c_1 c_2]
\ee
with $c_l=\cos\alpha_l$ and $C$ being a diagonal 3x3 matrix. 
Under the action of the channel $\cE_1$ a state
$\vec{\varrho}=\frac{1}{2}(I+\vec{r}\cdot\vec{\sigma})$ 
($\vec{r}=(z,y,z)$ and $|\vec{r}|\leq 1$) is transformed
into a state
$\varrho^{\,\prime}=\frac{1}{2}(I+\vec{r}^{\,\prime}\cdot\vec{\sigma})$
with $\vec{r}^{\,\prime}=R_2CR_1\vec{r}$ and $R_1,R_2$ being orthonogal
rotation matrices associated with the unitary conjugations $V_1,V_2$, 
respectively. As a result of QPT we obtain the channel $\cE_1$, i.e. 
the matrix $E_1=R_2 C R_1$. In order
to specify the parameters of $R_1$,$R_2$ and $C$ we 
implement the singular value decomposition of the 
reconstructed matrix $E_1$ (recall that by definition
the entries of the matrix $C$ are non-negative). After
this is completed it remains to determine the sign 
of $\alpha_z$ and $W_2$ that have no effect on $\cE_1$.

The unitary $W_2$ will manifest itself if we look at second concatenation
of the memory channel. In particular, we will focus on the 
conditional channel $\cE_{|\varrho_1}$ describing the transformation 
of the states following the input test state $\varrho_1$, i.e.
we apply QPT to learn the average mapping
\be
\cE_{|\rho_1}(\rho)&:=&\tr_{\hs{M},1}[{\cal{U}}_2(\overline{\xi}\otimes\rho_1\otimes\rho)]\nonumber\\
&=&V_2\tr_{\hs{M}}[D(C_{\rho_1}(\overline{\xi})\otimes V_1\rho V_1^{\dagger})D^{\dagger}]V_2^{\dagger},\nonumber
\ee 
where, ${\cal{U}}_2$ is the two-fold concatenation of the memory channel,
$C_{\rho_1}(\overline{\xi})=W_2\tr_1[D(\overline{\xi}\otimes V_1\rho_1V_1^{\dagger})D^{\dagger}]W_2^{\dagger}$ and we used $D\equiv D(\alpha)$.

Since $\overline{\xi}=\frac{1}{2}I$ we can easily compute $C_{\rho_1}(\frac{1}{2}I)=\frac{1}{2}(I + \vec{m}\cdot\vec{\sigma})$ where $\vec{m}=O_2SR_1\vec{r}_1$, $S={\rm diag}[s_2s_3,s_3s_1,s_1s_2]$ with $s_j=\sin\alpha_j$, $O_2$ corresponds to unitary rotation $W_2$, and we write $\rho_1=\frac{1}{2}(I +\vec{r}_1\cdot\vec{\sigma})$. In Bloch sphere representation we obtain
$\vec{r}\mapsto\vec{r}^{\,\prime}=R_2FR_1\vec{r}+R_2 S\vec{m}$, where 
\be
F=\left(
\ba{ccc}
c_2c_3 & m_zc_2s_3 & -m_yc_3s_2\\
-m_zc_1s_3 & c_3c_1 & m_xc_3s_1\\
m_yc_1s_2 & -m_xc_2s_1 & c_1c_2\\
\ea
\right).
\ee
As a result of the QPT we find the transformation
$\vec{r}\mapsto\vec{r}^{\,\prime}=T\vec{r}+\vec{t}$.
From the knowledge of $\vec{t}$ we can fully recover all 
the parameters of $W_2$ (via $O_2$) and the parameter 
$F_{21}$ (constructed from $T$) we can set the sign of 
$\alpha_z$.

\emph{Non-unique fixed point.}
In this case the complete mixture $\frac{1}{2}I$ is not the only 
fixed point of $\cal{C}$. We will show that for the considered case
this means the memory channel is driven by controled unitary evolution
we discussed before. By definition
\be
\nonumber
{\cal{C}}(\xi)=W_2\tr_1[D(\xi\otimes \frac{1}{2}\id)D^{\dagger}]W_2^{\dagger}
=W_2\cE_D(\xi)W_2^{\dagger}\,,
\ee
where $\cE_D$ is the as in Eq.\eqref{eq:superoper}, because $D$ is invariant with respect to exchange of the memory and the input system. The only possibility for $\cal{C}$ to have multiple fixed point is that the matrix $C$ has at least one eigenvalue $1$ and that the corresponding vector is left invariant under the action of $W_2$. This means that $c_k=c_l=\pm 1$ for two different $k,l$. Due to restrictions on $\alpha_k$ (see Eq.\eqref{eq:Dalpha}) we can immediately state that $\alpha_z=\alpha_y=0$, thus, $D=\cos(\alpha_x/2)\id\otimes\id+i\sin(\alpha_x/2)\sigma_x\otimes\sigma_x$ and $W_2=e^{i\beta\sigma_x}$. Let $P_{\pm}$ be the projections onto eigenvectors of $\sigma_x$ associated with the eigenvalues $\pm 1$, respectively. Then clearly $U=(W_2\otimes V_2)D(I\otimes V_1)=
e^{i\beta}P_+\otimes V_++e^{-i\beta}P_-\otimes V_-$, where
$V_\pm=V_2 e^{\pm i\frac{\alpha_x}{2} \sigma_x}V_1$,
is a controlled unitary operator, thus, by Theorem \ifarxi \ref{theorem:control_u} \else 1 (see the main text)\fi 
the reconstruction obtained by QPT results in one of the unitary 
channels induced by unitary operators $V_\pm$.

%% file: esterarxi.bbl
\begin{thebibliography}{10}
\bibitem{silent_swap} Notice that in this form we omit the swapping of memory and system after the interaction for the sake of clarity. Hence as written in text the interaction is ${\cal U}: \mm\otimes \hh\mapsto \mm\otimes \hh$ and not $\mm\otimes \hh\mapsto \hh\otimes \mm$, which is assumed when we need it for nice concatenation properties.
\end{thebibliography}
